\documentclass{article}

\usepackage{PRIMEAIarxiv}

\usepackage{graphicx}
\usepackage{natbib}
\usepackage{amsmath,amssymb}
\usepackage{amsthm}
\usepackage{mathtools}
\usepackage{algorithm}
\usepackage{algorithmic}
\makeatletter
\newcommand{\removelatexerror}{\let\@latex@error\@gobble}
\makeatother
\removelatexerror
\usepackage{algorithmic}
\usepackage{glossaries}

\newtheorem{assumption}{Assumption}
\newtheorem{definition}{Definition}

\newtheorem{theorem}{Theorem}
\newtheorem{proposition}{Proposition}
\newtheorem{remark}{Remark}

\newacronym{cpjmp}{CPJMP}{Controlled Pure Jump Markov Process}

\begin{document}

\title{Timing-Aware Two-Player Stochastic Games with Self-Triggered Control}

\author{Yunian Pan \& Quanyan Zhu \\
ECE Department\\
  New York University \\
  Brooklyn, NY\\
  \texttt{\{yp1170, qz494\}@nyu.edu}}

  \maketitle 
  
\begin{abstract}
We study self-triggered two-player stochastic games on Piecewise Deterministic Markov Processes (PDMPs) where each agent decides when to observe and which open-loop action to hold. Augmenting the state with clocks and committed controls yields flow regions (both hold) and trigger surfaces (at least one updates). The framework covers both blind simultaneous (Nash) timing and observable sequential (Stackelberg) commitments;
 the former leads to coupled, intractable QVIs, while the latter admits a nested
  Hamilton–Jacobi–Bellman quasi-variational inequality and a tractable dynamic-
  programming decomposition. We outline a computational scheme based on implicit
  differentiation of the follower’s fixed point. A pursuit–evasion example
  illustrates the strategic timing interaction.
\end{abstract}



\section{Introduction}

Self-triggered control has been as a fundamental paradigm for resource-efficient implementation of feedback control systems \citep{Heemels2012, AntaTabuada2010}. Unlike periodic sampling, self-triggered controllers determine the next sampling instant based on the current system state, enabling adaptive reduction of communication, computation, and actuation costs while maintaining performance guarantees. This flexibility is critical in energy-constrained cyber-physical systems (CPS), wireless sensor networks, and transportation networks~\cite{pan2023resilience}.

However, existing self-triggered control theory predominantly addresses benign stochastic disturbances or worst-case bounded errors under cooperation assumptions. In adversarial environments, including traffic assignment~\citep{pan2023stochastic,pan2022poisoned}, backdoor attack machine learning defense~\citep{pan2023first}, and competitive resource allocation, the system faces an intelligent opponent who can exploit the timing and parametrization of control updates. This creates a fundamentally game-theoretic problem where the defender must optimize not just \emph{what} to do, but \emph{when} to act, against a strategic adversary who operates under similar timing constraints.

This paper establishes a rigorous game-theoretic framework for \emph{adversarial self-triggered control} in continuous-time stochastic systems. We model the dynamics as Controlled Piecewise Deterministic Markov Processes (PDMPs), a general class encompassing hybrid systems where continuous flow is interrupted by stochastic or controlled jumps. Both the defender and the adversary operate in self-triggered mode, selecting inter-execution times and control parameters at discrete instants. This leads to a lifted state space that includes both players' physical commitments and residual clocks. We analyze \emph{Nash} (simultaneous) and \emph{Stackelberg} (leader-follower) information patterns, identifying the geometric structure of the "Game on Clocks" and providing a computational pathway to solve the resulting bilevel optimization problem. Our main contributions are: 
    (i) We formalize the adversarial self-triggered control problem as a two-player non-zero-sum game on PDMPs, and derive the associated Hamilton-Jacobi-Bellman (HJB) systems, characterizing the solution as a fixed point of coupled intervention operators.
    (ii) For the sequential (Stackelberg) information pattern, we prove the existence and uniqueness of the follower's response via a contraction mapping argument on the augmented state space. We further cast the leader's problem as a mathematical programming with equilibrium constraint.
    (iii)  To overcome the intractability of the coupled PDE systems, we develop a numerical scheme based on bi-level Approximate Dynamic Programming (ADP). We introduce implicit differentiation to estimate the hypergradient of the leader's objective through the follower's equilibrium constraints, enabling efficient gradient-based optimization of the self-triggered policy.

\subsection{Related Work}
\subsubsection{Self-Triggered and Event-Triggered Control}
The foundational work of \citet{AntaTabuada2010} established self-triggered control for deterministic nonlinear systems, where the next trigger time is pre-computed to guarantee Lyapunov stability. Extensions to stochastic systems \citep{Donkers2012, Molin2015} primarily address stabilization and $\mathcal{L}_2$-gain performance under exogenous disturbances. However, these works typically model disturbances as passive noise or worst-case signals with bounded variation. Our work departs from this literature by explicitly modeling the disturbance as a \emph{rational player} with their own self-triggering clock, leading to a "dual-clock" hybrid system where timing commitments are strategic variables.
Recent work on security in event-triggered control \citep{Senejohnny2017} has explored denial-of-service (DoS) attacks. Our framework differs by modeling the adversary as influencing the system dynamics continuously through held control actions (actuator attacks) rather than merely jamming the communication channel, and by solving for the equilibrium policies rather than just robust stability margins.
\subsubsection{Differential Games and Impulse Control}
Classical differential game theory \citep{Basar1999} typically assumes both players act continuously. Our setting introduces a hybrid structure where players commit to piecewise-constant actions that persist over dwell intervals. This connects our work to \emph{Impulse Games} \citep{chikrii2007differential,basei2022nonzero} and strategic information acquisition \citep{huang2021pursuit}, where players choose intervention times to reset the state. However, standard impulse games typically involve instantaneous state jumps (e.g., inventory replenishment). In contrast, our formulation involves ``parameter jumps'' (switching the vector field) and "clock resets," which creates a distinct transport-equation geometry in the HJB formulation. 
Furthermore, while stochastic games \citep{Filar1997} often assume symmetric discrete-time actions, our Stackelberg formulation leverages the time-scale separation inherent in self-triggering. This relates to recent advances in \emph{Feedback Stackelberg Equilibrium} for dynamic games, which we extend to the PDMP setting using implicit gradient estimation.

\section{Problem Formulation}
\label{sec:formulation}

We consider a class of stochastic systems modeled as Controlled Piecewise Deterministic Markov Processes (PDMP) on a general state space. The system is subject to two opposing inputs: a control input $u_1$ applied by Player~1 (the ``Defender''), and a disturbance or adversarial input $u_2$ applied by Player~2 (the ``Adversary''). Both players determine when to observe and refresh their open-loop actions by selecting self-triggered dwell times.

\subsection{System Dynamics and Polish State Space}

Let $(S, \mathcal{B}(S))$ be a Polish space representing the system states. The state $X(t) \in S$ evolves according to a PDMP characterized by a tuple $(S, \lambda, Q)$.

\textbf{Dynamics.}
{\it
For any fixed control inputs $u_1 \in A_1$ and $u_2 \in A_2$ (where $A_i$ are compact metric spaces), the evolution of the state is generated by the infinitesimal operator $\mathcal{L}^{u_1, u_2}$ acting on the domain of test functions $\phi \in \mathcal{D}(\mathcal{L})$:}
\begin{equation}
\label{eq:generator}
\begin{aligned}
    \mathcal{L}^{u_1, u_2} \phi(x) &  = \underbrace{f(x, u_1, u_2) \cdot \nabla \phi(x)}_{\text{Continuous Flow}}  \\ 
    & + \underbrace{\lambda(x, u_1, u_2) \int_S (\phi(y) - \phi(x)) Q(x, dy; u_1, u_2)}_{\text{Stochastic Jump}}.
\end{aligned}
\end{equation}
Here, $f$ is the drift vector field, $\lambda$ is the jump intensity function, and $Q(\cdot | x, u)$ is the post-jump transition kernel. We assume these functions satisfy standard Lipschitz and boundedness conditions to ensure the process is well-defined in the sense of probability law.

In this paper, we formulate the system on a specific subset of self-triggering mechanisms.

\begin{definition}
The players updates the control at discrete triggering instants $\{\tau_{i,k}\}_{k \ge 0}$. At each $\tau_{i,k}$, they observe the system (subject to the information pattern) and commit to:
\begin{enumerate}
    \item An \emph{inter-execution time} $T_{i,k} = \tau_{i,k+1} - \tau_{i,k} \in [\underline{T}_i, \overline{T}_i]$.
    \item A \emph{control parameter} $\theta_{i,k} \in \Theta_i$ defining the open-loop control trajectory $u_i(t) = \Gamma_i(t - \tau_{i,k}; \theta_{i,k})$ for the interval $[\tau_{i,k}, \tau_{i,k+1})$.
\end{enumerate}

\end{definition}

We consider a general class (not necessarily zero-sum) of performance functionals associated with the self-triggering setting.

\begin{definition}[Performance Functional]
The defender and the adversary $i \in \{1,2\}$ seeks to minimize, the infinite-horizon discounted costs:
\begin{equation}
\label{eq:cost}
\begin{aligned}
    J(x, \pi_1, \pi_2)&  = \mathbb{E}_x \bigg[
\int_0^\infty e^{-\gamma t} r_i\big(X(t), u_1(t), u_2(t)\big) dt
\\ 
& \qquad \quad + \sum_{k=0}^\infty e^{-\gamma \tau_{i, k}} g_i(\tau_{i, k+1} - \tau_{i, k})
\bigg], 
\end{aligned}
\end{equation}
where constant $\gamma > 0$ is the discount rate, $r_i$ and $g_i$ are the running cost functionals for state/controls and observations.

\end{definition}

\subsection{Augmented State Space and Information Structure}

The solution to the self-triggered control problem depends fundamentally on the \emph{information structure} at the triggering instant. 
To preserve \emph{Markovian property} of the self-triggered operations, we allow players to track the physical state as well as opponent's timing commitment and held control. We therefore lift the dynamics to
$\chi(t) = \big(x(t), \sigma_1(t), \theta_1(t), \sigma_2(t), \theta_2(t)\big) \in \mathcal{X} := S \times [0, \overline{T}_1] \times \Theta_1 \times [0, \overline{T}_2] \times \Theta_2,
$
where:
\begin{itemize}
    \item $x(t) \in S$ is the physical PDMP state.
    \item $\sigma_i(t) = \tau_{i,k_i+1} - t \in [0, \overline{T}_i]$ is Player $i$'s residual time until the next self-trigger (clock variable).
    \item $\theta_i(t) = \theta_{i,k_i} \in \Theta_i$ is the parameter defining Player $i$'s held open-loop control between triggers.
\end{itemize}
Between triggering instants both clocks decrease deterministically, $\dot{\sigma}_i(t) = -1$, while the control parameters remain constant. When Player $i$'s clock hits zero, they observe the physical state, select a new dwell time $T_i$ and parameter $\theta_i$, and reset $\sigma_i(\tau_{i,k_i}^+) = T_i$, $\theta_i(\tau_{i,k_i}^+) = \theta_i$.

\begin{definition}
\label{def:admissible}
Under perfect information structure, player $i$'s self-triggered policy is a measurable mapping $\pi_i: \mathcal{X} \to [\underline{T}_i, \overline{T}_i] \times \Theta_i$ that assigns $(T_i, \theta_i) = \pi_i(\chi)$ at each of its triggering instants. Once $(T_i, \theta_i)$ is chosen, the held control is $u_i(t) = \Gamma_i(t-\tau_{i,k_i}; \theta_i)$ for $t \in [\tau_{i,k_i}, \tau_{i,k_i}+T_i)$. The collection of admissible policies is denoted $\Pi_i$.
\end{definition}

\begin{remark}
    Under perfect information structure, at the decision timing, a player essentially only needs $x(t)$, $\sigma_{-i}(t)$, and $\theta_{-i}(t)$ as their previous commitment $\theta$ expires and their old clock $\sigma_i(t)$ is bound to be reset.
\end{remark}

\begin{assumption}[Regularity]
\label{ass:regularity}
The primitives satisfy:
\begin{enumerate}
    \item $r(\cdot)$ and $g(\cdot)$ are bounded, jointly continuous, and Lipschitz in the state argument.
    \item $\Gamma(\cdot; \theta)$ is measurable in $s$, continuous in $\theta$, and uniformly bounded for $\theta \in \Theta$, where $\Theta$ is compact.
    \item The drift $f$, intensity $\lambda$, and post-jump kernel $Q$ are jointly continuous and satisfy both linear growth and Lipschitz conditions that ensure the PDMP admits a unique strong solution under any measurable controls.
\end{enumerate}
\end{assumption}

\section{Solutions and Characterizations}
\subsection{Simultaneous Move Information Pattern (Nash)}
\label{sec:nashsolution}

In the simultaneous move information pattern, we assume that at any instant $t$, the players' strategies depend only on the current augmented state $\chi(t)$. We seek a pair of feedback Nash equilibrium policies $(\pi_1^*, \pi_2^*)$ such that each is optimal against the other. This leads to a system of coupled HJB equations.

The augmented state is partitioned into the flow region
$\Omega_{\text{flow}} = \{ \chi : \sigma_1 > 0, \sigma_2 > 0 \},$ and boundaries:  $\Gamma_i = \{\chi : \sigma_i = 0, \sigma_{-i} \geq 0\}$ where at least one of the players' clock expires. 
Information structures specify what is observed on each boundary (e.g., whether Player~2 sees Player~1's commitment before choosing its own). 
Let $V_1(\chi)$ and $V_2(\chi)$ denote the value functions for the Defender and Adversary, respectively. Due to the self-triggered nature of the control, the control parameters $\theta_i$ are fixed while $\sigma_i > 0$. The decision-making is confined to the boundaries $\Gamma_i$.

To facilitate the dynamic programming formulation, we define the \emph{Intervention Operators} $\mathcal{M}_i$. These operators characterize the optimization problem solved by a player when their clock expires.
\begin{definition}[Intervention Operators]

For any test function $\phi: \mathcal{X} \to \mathbb{R}$ and state $\chi \in \Gamma_i$ (where $\sigma_i = 0$), the operator $\mathcal{M}_i$ is defined as:\begin{equation}\mathcal{M}_i \phi = \inf_{(T_i, \theta_i) \in [\underline{T}_i, \overline{T}_i] \times \Theta_i} \left\{ g_i(T_i) + \phi(x, \hat{\sigma}_1, \hat{\theta}_1, \hat{\sigma}_2, \hat{\theta}_2) \right\},\end{equation}where the post-reset state components are updated as follows: $\hat{\sigma}_i = T_i$, $\hat{\theta}_i = \theta_i$, and for the non-acting player $j \neq i$, $\hat{\sigma}_j = \sigma_j, \hat{\theta}_j = \theta_j$.
\end{definition}

\begin{remark}[Ordering of Events]
While the operator $\mathcal{M}_i$ is defined uniformly, the subsequent evolution of the system depends on the relationship between the chosen dwell time $T_i$ and the opponent's residual time $\sigma_{-i}$. The time to the next decision epoch is $\tau_{next} = \min(T_i, \sigma_{-i})$.
\begin{itemize}
    \item If $T_i < \sigma_{-i}$, the trajectory returns to $\Gamma_i$, implying Player $i$ acts two consecutive times.
    \item If $T_i \ge \sigma_{-i}$, the trajectory hits $\Gamma_{-i}$ next, meaning Player $-i$ will interrupt Player $i$'s interval to update their own control.
\end{itemize}
This switching logic is implicitly captured by the viscosity solution of the HJB equation in the flow region $\Omega_{\text{flow}}$.
\end{remark}

The Feedback Nash equilibrium value functions $(V_1, V_2)$ must satisfy the following system of coupled quasi-variational inequalities. The system is ``nested'' in the sense that the boundary condition of one region serves as the initialization for the flow in the other.

\subsubsection{The Flow Region $(\Omega_{\text{flow}})$}
When $\sigma_1 > 0$ and $\sigma_2 > 0$, no decisions are made. The system evolves according to the PDMP dynamics and the deterministic decay of the clocks. For $i \in \{1, 2\}$:
\begin{equation}
\label{eq:HJB_flow}
\gamma V_i(\chi) + \frac{\partial V_i}{\partial \sigma_1}(\chi) + \frac{\partial V_i}{\partial \sigma_2}(\chi) - \mathcal{L}^{\theta_1, \theta_2} V_i(\chi) - r_i(x, \theta_1, \theta_2) = 0.
\end{equation}
Here, the transport terms $\frac{\partial}{\partial \sigma}$ arise from the clock dynamics $\dot{\sigma} = -1$.

\subsubsection{The Triggering Boundaries $(\Gamma_i)$}
When player $i$'s clock expires ($\sigma_i = 0$), they must act optimally given the current state of the opponent. More specifically, on boundaries $\Gamma_i$
\begin{equation}
\label{eq:HJB_bound1}
\begin{aligned}
    V_i(\chi) &= \mathcal{M}_i [V_1](\chi), \\
    V_{-i}(\chi) &= V_{-i}(x, T_i^*, \theta_i^*, \sigma_{-i}, \theta_{-i}),
\end{aligned}
\end{equation}
where $(T_1^*, \theta_1^*)$ is the argument minimizing the RHS of the first equation. Note that opponent's value function undergoes a discrete jump in its arguments (due to Player $i$'s update) but incurs no immediate impulse cost.

\subsubsection{Simultaneous Expiry $(\Gamma_1 \cap \Gamma_2)$}
In the rare event where $\sigma_1 = \sigma_2 = 0$, the players engage in a static game at the boundary. A Nash equilibrium pair $((T_1^*, \theta_1^*), (T_2^*, \theta_2^*))$ is chosen such that:
\begin{equation}
    (T_i^*, \theta_i^*) \in \arg\min_{(T_i, \theta_i)} \left\{ g_i(T_i) + V_i(x, T_1, \theta_1, T_2, \theta_2) \right\}.
\end{equation}
This constitutes a static game played on the value surface.

The solution to the coupled system \eqref{eq:HJB_flow} provides the Nash equilibrium strategies.

\begin{theorem}
Let $v_1, v_2 \in C^1(\Omega_{\text{flow}}) \cap C^0(\bar{\mathcal{X}})$ be a viscosity solution to the coupled system satisfying polynomial growth conditions. Then:
\begin{enumerate}
    \item $v_1$ and $v_2$ coincide with the value functions of the Feedback Nash game defined in \eqref{eq:cost}.
    \item The optimal self-triggered policies $\pi_i^*$ are generated by the arguments of the intervention operators $\mathcal{M}_i$ at the boundaries $\Gamma_i$.
\end{enumerate}
\end{theorem}

\begin{proof}[Proof Sketch]
We focus on Player 1, the argument for Player 2 is symmetric. Let $(v_1, v_2)$ be a solution to the coupled system \eqref{eq:HJB_flow}. Fix Player 2's strategy to the equilibrium policy $\pi_2^*$. Let the sequence of triggering instants for Player 1 be $\{\tau_k\}_{k\ge 0}$.

Consider the quantity $e^{-\gamma t} v_1(\chi(t))$. In the interval between triggers $t \in (\tau_k, \tau_{k+1})$, the state $\chi(t)$ evolves in the interior $\Omega_{\text{flow}}$. Applying Dynkin's formula (or the chain rule for PDMPs) gives:
\begin{equation*}
\begin{aligned}
     & \  \mathbb{E} \left[ e^{-\gamma \tau_{k+1}} v_1(\chi(\tau_{k+1}^-)) \right] - e^{-\gamma \tau_k} v_1(\chi(\tau_k^+)) \\
    & = \mathbb{E} \left[ \int_{\tau_k}^{\tau_{k+1}} e^{-\gamma s} \mathcal{L} [e^{\gamma s} e^{-\gamma s} v_1] ds \right].
\end{aligned}
\end{equation*}
Rearranging terms and substituting the HJB flow equation, we have:
    $e^{-\gamma \tau_k} v_1(\chi(\tau_k^+)) = $ $ \mathbb{E} \bigg[ \int_{\tau_k}^{\tau_{k+1}} e^{-\gamma s} r_1(\chi(s)) ds + e^{-\gamma \tau_{k+1}} v_1(\chi(\tau_{k+1}^-)) \bigg].$
This equation relates the value at the start of an interval to the running cost accumulated during flow plus the value at the end of the interval.

At a triggering instant $\tau_{k+1}$, the system hits a boundary.
\begin{itemize}
    \item \emph{Case A (Player 1 acts):} If $\chi(\tau_{k+1}^-) \in \Gamma_1$, the intervention operator $\mathcal{M}_1$ applies. By definition of the infimum in \eqref{eq:HJB_bound1}:
        $v_1(\chi(\tau_{k+1}^-)) \le g_1(T_1) + v_1(\chi(\tau_{k+1}^+)),$
    with equality holding if $T_1$ is chosen according to the optimal policy $\pi_1^*$.
    \item \emph{Case B (Player 2 acts):} If $\chi(\tau_{k+1}^-) \in \Gamma_2$, Player 2 updates. Since $v_1$ is continuous across Player 2's reset (Player 1 pays no cost), we have $v_1(\chi^-) = v_1(\chi^+)$.
\end{itemize}
Combining these, for any generic policy $\pi_1$, we have the inequality:
    $v_1(\chi(\tau_{k+1}^-)) \le \text{Cost}_1(\text{jump at } k+1) + v_1(\chi(\tau_{k+1}^+)).$

Substituting the inequality back into the Dynkin's formula and summing over $k=0$ to $N-1$, we arrive at
    $v_1(\chi(0)) \le \mathbb{E} \bigg[ \int_0^{\tau_N} e^{-\gamma s} r_1 ds + \sum_{k=0}^{N-1} e^{-\gamma \tau_{k+1}} g_1(\Delta \tau) + e^{-\gamma \tau_N} v_1(\chi(\tau_N)) \bigg].$
Taking the limit as $N \to \infty$, and invoking the transversality condition (guaranteed by the discount factor $\gamma$ and boundedness of $v_1$), the terminal term $e^{-\gamma \tau_N} v_1 \to 0$.
Thus, $v_1(\chi(0)) \le J_1(\chi(0), \pi_1, \pi_2^*)$. Since this holds for any $\pi_1$, $v_1$ is a lower bound.
\end{proof}

This formulation displays the computational difficulty: the value function $V_1$ depends on the optimal strategy of Player 2 embedded in the boundary condition of $V_2$, and vice versa. Mathematically, this constitutes a system of coupled Quasi-Variational Inequalities (QVI). Unlike standard control problems where the boundary condition is static (e.g., $V(x)=0$ at target), the boundary values here are functional solutions to the opponent's optimization problem. Consequently, the pair $(V_1, V_2)$ must be found as a fixed point of the coupled intervention map, typically requiring iterative relaxation schemes that may suffer from convergence issues without strict contraction properties.

\subsection{Sequential Move Information Pattern (Stackelberg)}
\label{sec:stackelberg}

In the sequential formulation, we adopt the \emph{Feedback Stackelberg} information structure. The game proceeds in two stages:
\begin{enumerate}
    \item  Prior to the game, the leader (Defender) announces policy $\pi_1: \mathcal{X} \to [\underline{T}_1, \overline{T}_1] \times \Theta_1$. This policy is binding and known to Player 2 for the entire infinite horizon.
    \item  the follower (Adversary) observes $\pi_1$ and acts as a rational optimal controller, solving for the best-response policy $\pi_2^*(\cdot; \pi_1)$ that minimizes their own cost.
\end{enumerate}

This structure decouples the game into an inner optimal self-triggered control problem for the follower and an outer optimization over policy space for the leader.

Given the fixed leader policy $\pi_1$, the adversary views Player 1's behavior not as a strategic uncertainty, but as a known state-dependent drift in the environment. 
Let $V_2^{\pi_1}(\chi)$ denote the adversary's value function conditioned on the leader's policy $\pi_1$. It satisfies the following HJB system:
\begin{enumerate}
    \item The local evolution in flow region $\Omega_{\mathrm{flow}}$ is governed by the held controls:
\begin{equation}\label{eq:stackelbergflow}
    \gamma V_2^{\pi_1} + \frac{\partial V_2^{\pi_1}}{\partial \sigma_1} + \frac{\partial V_2^{\pi_1}}{\partial \sigma_2} - \mathcal{L}^{\theta_1, \theta_2} V_2^{\pi_1} - r_2(x, \theta_1, \theta_2) = 0.
\end{equation}
\item  At boundaries, Player 2 either acts optimally to minimize their cost-to-go, or lets the transition be dictated by the leader's new commitment:
\begin{equation}\label{eq:stackelbergint}
\begin{aligned}
       V_2^{\pi_1}(\chi) & = \mathcal{M}_2 [V_2^{\pi_1}](\chi). \\
        V_2^{\pi_1}(\chi)& = V_2^{\pi_1}\big(x, T_1^{\pi}, \theta_1^{\pi}, \sigma_2, \theta_2\big).
\end{aligned}
\end{equation}
\end{enumerate}

The solution to this system yields the unique best-response policy $\pi_2^*(\cdot; \pi_1)$, defined by the argument minimizing the RHS of the intervention operator $\mathcal{M}_2$.

The leader seeks the optimal policy $\pi_1$ that minimizes their value function, anticipating the follower's optimal response, hence formulating a bilevel programming:
\begin{equation}
\label{eq:stackelberg_leader}
\begin{aligned}
    \min_{\pi_1 \in \Pi_1} \quad & J_1\big(x_0, \pi_1, \pi_2^*(\cdot; \pi_1)\big) \\
    \text{s.t.} \quad & \pi_2^*(\cdot; \pi_1) \text{ solves \eqref{eq:stackelbergflow}--\eqref{eq:stackelbergint} }
\end{aligned}
\end{equation}

\begin{remark}
A practical realization of this framework is Moving Target Defense (MTD). The leader commits to a configuration rotation policy $\pi_1$ (e.g., an IP hopping algorithm). The Adversary (Follower) observes this logic and selects a reconnaissance duration $T_2$ to identify vulnerabilities. Since the Adversary is physically committed to the scan for time $T_2$, the Defender's optimal policy $\pi_1$ exploits this latency, ensuring the system state rotates (invalidating the gathered intelligence) exactly as the Adversary's clock expires ($\sigma_2 \to 0$).
\end{remark}

\section{Dynamic Programming Principles }
\label{sec:dp}

In this section, we establish the Dynamic Programming Principle (DPP) for the self-triggered game. We reformulate the differential HJB systems as fixed-point problems involving integral operators. This operator-theoretic perspective provides the necessary conditions for the existence and uniqueness of solutions.

Let $\mathbb{M}(\bar{\mathcal{X}})$ denote the space of bounded measurable functions on the closed state space, equipped with the supremum norm $\|v\|_\infty = \sup_{\chi \in \bar{\mathcal{X}}} |v(\chi)|$.
We define $\mathcal{P}_t^{\theta}$ as the \emph{Markov Transition Semigroup} associated with the PDMP generator $\mathcal{L}^{\theta_1, \theta_2}$ (with fixed controls between triggers). For any bounded measurable function $v$ and time $t \ge 0$, the operator is given by the expected cost-to-go over the interval $[0, t]$:
\begin{equation*}
    \mathcal{P}_t^{\theta}[v](\chi) = \mathbb{E}_\chi \left[ \int_0^t e^{-\gamma s} r(X(s), \theta_1, \theta_2) ds + e^{-\gamma t} v(\chi(t)) \right],
\end{equation*}
where the expectation is taken with respect to the law of the PDMP starting at $\chi$ with constant controls $\theta$.
The essential ``step size'' in our self-triggered formulation is the time until the \emph{next} boundary hit, defined as: $\tau(\chi) = \min \{ \sigma_1, \sigma_2 \}$.

For the simultaneous move game, the value functions $V_1, V_2$ are coupled. We define the Nash Bellman operator $\mathfrak{N}: \mathbb{M}(\bar{\mathcal{X}})^2 \to \mathbb{M}(\bar{\mathcal{X}})^2$.
The operator acts on a pair of candidate value functions $\mathbf{v} = (v_1, v_2)$. For a state $\chi$, the operator evaluates the evolution until $\tau(\chi) = 0$, at which point the intervention operators $\mathcal{M}_i$ (defined in Sec. \ref{sec:nashsolution}) are applied.

\begin{definition}
The Nash operator $\mathfrak{N}[\mathbf{v}] = (\mathfrak{N}_1[\mathbf{v}], \mathfrak{N}_2[\mathbf{v}])$ is defined as:
\begin{equation*}
    \mathfrak{N}_i[v_1, v_2](\chi) = 
    \begin{cases}
        \mathcal{P}_{\tau(\chi)}^{\theta}[v_i](\chi) & \text{if } \chi \in \Omega_{\text{flow}}, \\
        \mathcal{M}_i[v_i](\chi) & \text{if } \chi \in \Gamma_i \setminus \Gamma_{-i}, \\
        \text{Val}_i\big( \texttt{Static}(v_1, v_2) \big) & \text{if } \chi \in \Gamma_1 \cap \Gamma_2.
    \end{cases}
\end{equation*}
\end{definition}

A pair of functions $(V_1, V_2)$ constitutes a Feedback Nash Equilibrium solution if and only if they are a fixed point of the operator $\mathfrak{N}$:
\begin{equation}
    (V_1, V_2) = \mathfrak{N}[V_1, V_2].
\end{equation}
The Nash operator $\mathfrak{N}$ is \emph{not} a contraction mapping in general due to the coupling at the boundaries, and that the dwell time optimization can be generally nonconvex.

\subsubsection{Stackelberg Operator Theory}
In the sequential formulation, the operator structure simplifies due to the hierarchical decoupling.
Given a fixed leader policy $\pi_1$, we define the follower's Bellman operator $\mathfrak{S}_{\pi_1}: \mathbb{M}(\bar{\mathcal{X}}) \to \mathbb{M}(\bar{\mathcal{X}})$.
\begin{equation}
    \mathfrak{S}_{\pi_1}[v](\chi) = 
    \begin{cases}
        \mathcal{P}_{\tau(\chi)}^{\theta}[v](\chi) & \text{if } \chi \in \Omega_{\text{flow}}, \\
        \inf_{T_2, \theta_2} \left\{ g_2(T_2) + v(\chi^+_{\text{reset}}) \right\} & \text{if } \chi \in \Gamma_2, \\
        v(\text{Update via } \pi_1(\chi)) & \text{if } \chi \in \Gamma_1.
    \end{cases}
\end{equation}

\begin{theorem}
\label{thm:contraction}
For a fixed $\pi_1$, if the discount rate $\gamma > 0$ and dwell times are bounded away from zero ($T_{min} > 0$), the operator $\mathfrak{S}_{\pi_1}$ is a contraction mapping on $\mathbb{M}(\bar{\mathcal{X}})$ with modulus $\beta = e^{-\gamma T_{min}}$.
\begin{enumerate}
    \item There exists a unique value function $V_2^{\pi_1}$ such that $V_2^{\pi_1} = \mathfrak{S}_{\pi_1}[V_2^{\pi_1}]$.
    \item The value iteration $v^{(k+1)} = \mathfrak{S}_{\pi_1}[v^{(k)}]$ converges to $V_2^{\pi_1}$ geometrically.
\end{enumerate}
\end{theorem}

\begin{proof}[Proof Sketch]
We verify the contraction property using the discount factor accumulated over decision epochs. Let $\Delta \tau$ denote the time elapsed between two consecutive boundary hits.
Consider two candidate functions $u, v$. By the properties of the infimum and the linearity of the integral operator:
$$ \|\mathfrak{S}_{\pi_1}[u] - \mathfrak{S}_{\pi_1}[v]\|_\infty \le \sup_{\chi} \mathbb{E} \left[ e^{-\gamma \Delta \tau} \right] \|u - v\|_\infty. $$
While $\Delta \tau$ can be arbitrarily small in a single step (if $\sigma_{-i} \to 0$), the self-triggered admissibility condition requires that any reset sets the local clock to at least $\underline{T} > 0$. Consequently, for any trajectory, the cumulative time elapsed after at most two consecutive jumps (one by Player 2, one by Player 1) is lower-bounded by $\min(\underline{T}_1, \underline{T}_2)$. 
Therefore, the second iterate of the operator, $\mathfrak{S}_{\pi_1}^2$, is a strict contraction with modulus $\beta = e^{-\gamma \min(\underline{T}_1, \underline{T}_2)} < 1$. By the generalized Banach fixed point theorem, this guarantees the existence and uniqueness of the solution $V_2^{\pi_1}$. Hence the value iteration convergence is immediate.
\end{proof}

The leader's problem can now be stated as an optimization over the fixed points of the follower's operator. Let $\mathcal{F}: \Pi_1 \to \mathbb{M}(\bar{\mathcal{X}})$ be the map that assigns the unique fixed point to a policy: $\mathcal{F}(\pi_1) := \text{fix}(\mathfrak{S}_{\pi_1})$.
Then, the leader's DPP is to find:
$\pi_1^* = \arg \inf_{\pi_1 \in \Pi_1} J_1(x_0, \pi_1, \mathcal{F}(\pi_1)).$

\subsubsection{Approximate Dynamic Programming (ADP)} 
We therefore propose a numerical scheme based on ADP. We parameterize the leader's policy $\pi_1(\cdot; \xi)$ and the follower's value function $V_2(\cdot; \omega)$ using function approximators (e.g., neural networks) with weights $\xi \in \mathbb{R}^{d_\xi}$ and $\omega \in \mathbb{R}^{d_\omega}$. Finding a Stackelberg equilibrium then requires computing the \emph{hypergradient} of the leader's objective $\nabla_\xi \mathcal{J}_1 := J(\xi, \mathcal{F}(\xi))$, which captures how the follower's equilibrium shifts in response to the leader's policy.

\begin{proposition}[Implicit Differentiation]
\label{prop:implicit_gradient}
Let $\omega^*(\xi)$ be a fixed point of the follower's Bellman operator $\mathfrak{S}_\xi$.
Assume the strong regularity holds at the solution (i.e., the follower's optimal policy is unique and strict complementarity holds). 
Then, the mapping $\xi \mapsto \omega^*(\xi)$ is locally differentiable. The gradient of the leader's objective is given by:
\begin{equation}
\label{eq:hypergradient}
    \nabla_\xi \mathcal{J}_1 = \frac{\partial J_1}{\partial \xi} + \frac{\partial J_1}{\partial \omega} \left( I - \partial_\omega \mathfrak{S}_\xi \right)^{-1} \partial_\xi \mathfrak{S}_\xi,
\end{equation}
where $\partial$ denotes the Jacobian matrix.
\end{proposition}

We can thereby adopt a bi-level iterative scheme, derived from the sensitivity analysis in Proposition \ref{prop:implicit_gradient}, to approximately solve \eqref{eq:stackelberg_leader}. A typical procedure can be the following: 
\begin{enumerate}
    \item \textbf{Inner loop}: 
    fix leader configuration $\xi_k$, we solve for the follower's rational response. This entails finding the fixed point $\omega^*$ of the parametric Bellman operator:
  $ \omega^* = \mathfrak{S}_{\xi_k}[\omega^*].$
    This step is typically performed via fixed-point iteration (Value Iteration) or by minimizing the Bellman residual through actor-critic.
    \item \textbf{Outer loop}: 
    Once the inner loop converges to $\omega^*$, we compute the gradient of the leader's objective $\mathcal{J}_1(\xi) = J_1(\xi, \omega^*(\xi))$. Using the result from Proposition \ref{prop:implicit_gradient}.
    The leader then updates parameters using iterative schemes. (e.g.,  $\xi_{k+1} = \xi_k - \alpha_k \nabla_\xi \mathcal{J}_1,$
    where $\alpha_k$ is some proper step-size sequence. This process repeats until the convergence to a local minimum.
\end{enumerate}

\section{Numerical Illustration}
\label{sec:simulation}

We illustrate the strategic timing policy in a 2D pursuit-evasion game. The state is the relative position/velocity $z=[p-e,\dot p-\dot e]\in\mathbb{R}^4$ governed by a double integrator with held accelerations between triggers. 
The pursuer (follower) and the evader (leader) choose their dwell times $T_1,T_2$ and hold acceleration control in between triggers. They observe each other's clock and acceleration control at their own triggers.

Costs are quadratic in $z$ and control (standard LQ pursuit-evasion), with a constant trigger cost (no dwell penalty). The pursuer and evader's timing/control policies are learned with the augmented-state HJB/actor--critic described earlier, using fully connected networks on $\chi=[z,\sigma_1,\sigma_2,u_1,u_2]$.
Fig.~\ref{fig:trigger-policy} shows the pursuer’s trigger sensitivity to the
  evader’s clock. Dwell time shrinks as the opponent’s clock grows, i.e., the
  pursuer tries to replan while the evader is ‘asleep’; dwell also shortens with
  increasing distance, reflecting higher urgency when far. Control magnitude
  rises to saturation when the opponent’s clock is high, again exploiting periods
  when the evader has just committed.
\begin{figure}[htbp]
    \centering
    \includegraphics[width=\linewidth]{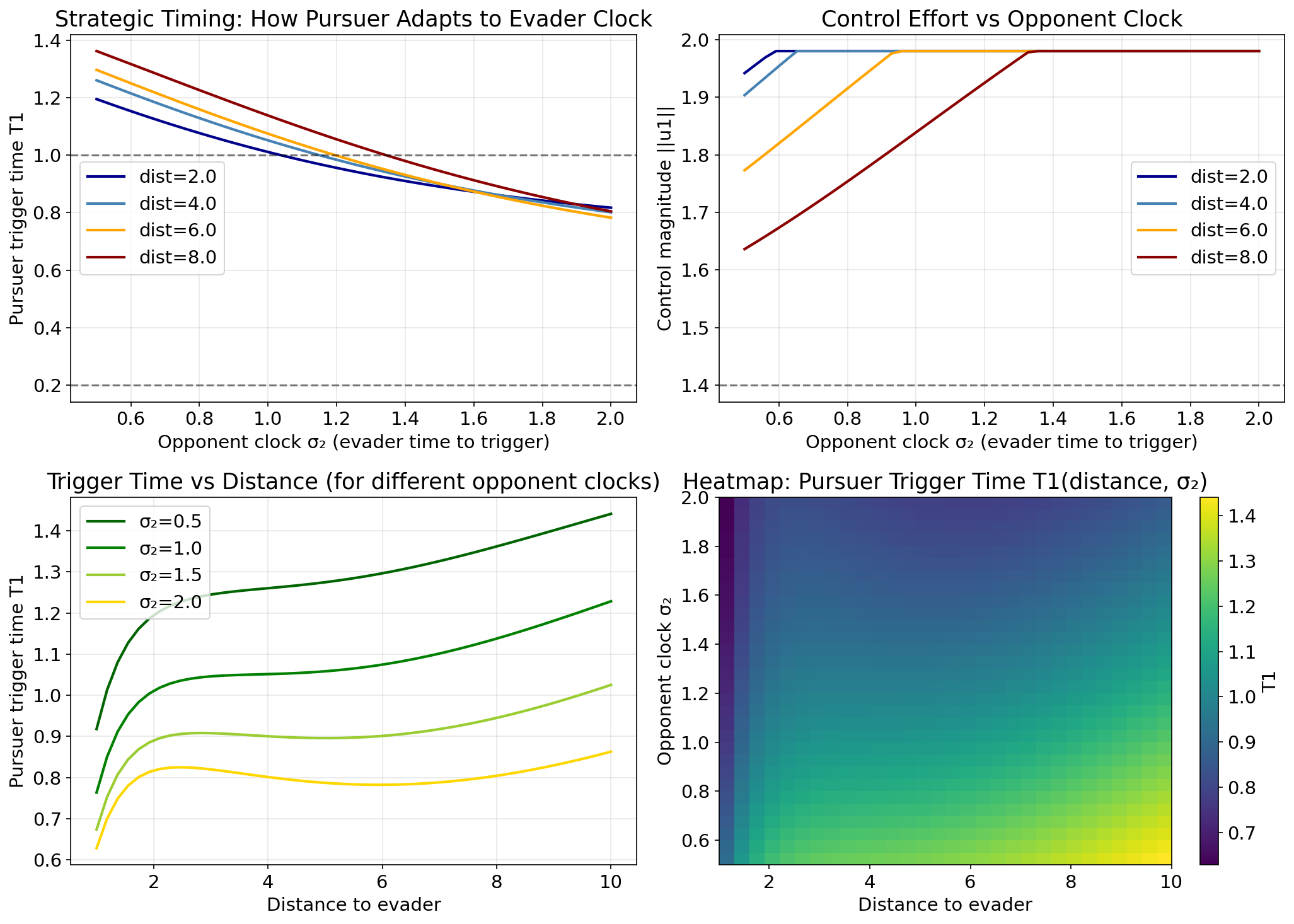}
    \caption{Trigger policy sensitivity of the learned pursuer vs. opponent clock $\sigma_2$ and distance.}
    \label{fig:trigger-policy}
\end{figure}

Fig.~\ref{fig:pursuit-comparison} compares the continuous LQ–Nash baseline (no
  triggers) with the learned Stackelberg, trigger-aware pursuer. The Stackelberg
  pursuer still achieves capture, but with additional triggers and a longer
  horizon when trigger costs are high, highlighting the trade-off between timing-
  aware updates and resource expenditure.

\begin{figure}[htbp]
    \centering
    \includegraphics[width=.6\linewidth]{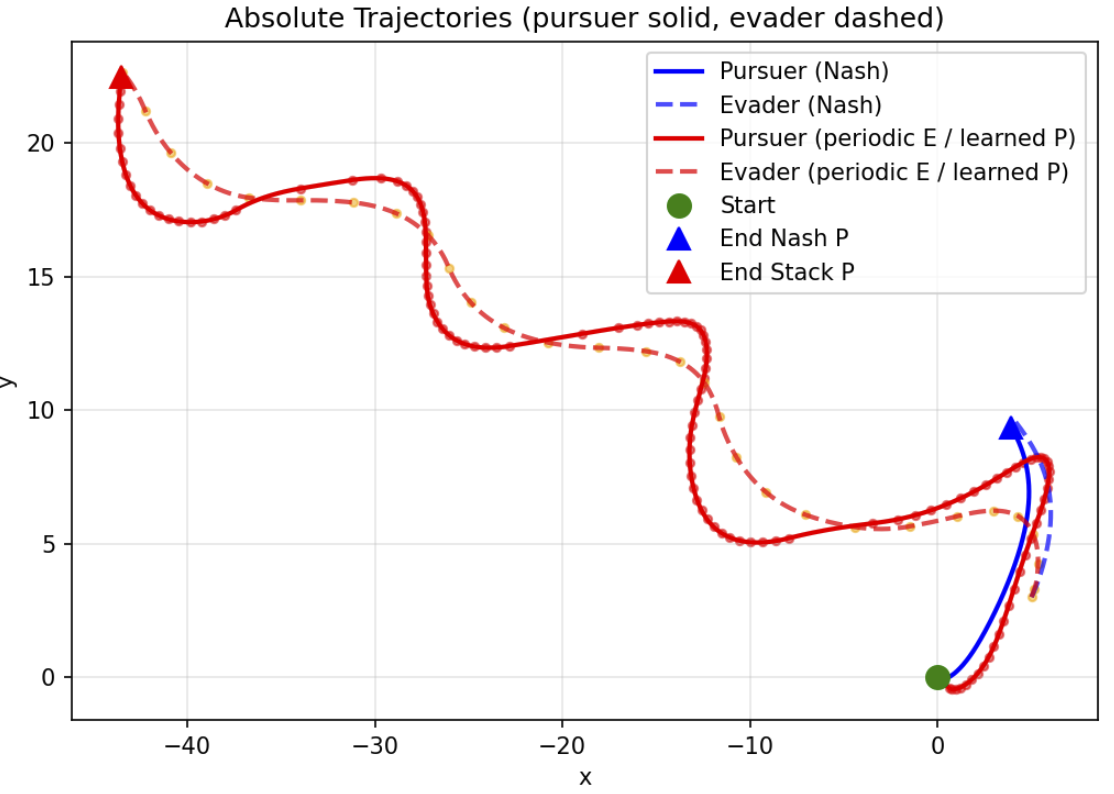}
    \caption{Learned pursuer (solid red) vs. evader (dashed red) compared to Nash solution (blue). Triggers are marked along the trajectory.}
    \label{fig:pursuit-comparison}
\end{figure}

\section{Conclusion}

We have developed a rigorous framework for self-triggered two-player games on PDMPs. By lifting the state to include dynamic clocks and committed controls, we characterized the hybrid geometry where players strategically trade off dwell-time costs against performance. Under the observable sequential commitment (Stackelberg) assumption, we derived a tractable dynamic programming principle, proving that the follower’s response constitutes a contraction mapping. The pursuit-evasion case study illustrated the strategic interaction between the players, revealing how the optimal policy adapts dwell times to exploit the adversary's ``lock-in'' periods.

Future research directions include addressing the blind-commitment Nash case via iterative relaxation methods, incorporating partial observability where the opponent's clock state must be inferred from physical measurements, and performing Large-scale empirical validation on high-fidelity CPS benchmarks, such as power grid frequency control or autonomous fleet coordination.

\bibliographystyle{abbrv}
\bibliography{ifacconf}

\end{document}